\documentclass{llncs}

\usepackage{cite}
\usepackage{hyperref}
\usepackage{times}
\usepackage{courier}
\usepackage{latexsym}
\usepackage{graphics}
\usepackage{comment}
\usepackage{xspace}
\usepackage{amssymb}
\usepackage{amsmath}
\usepackage{url}
\usepackage{alltt}
\usepackage{graphicx}
\usepackage{tabularx}
\usepackage{array} 
\usepackage{footnote} 
\usepackage{multirow}
\usepackage[table]{xcolor}
\usepackage{helvet}
\usepackage[latin1]{inputenc}
\usepackage[defblank]{paralist}
\usepackage{longtable}
\usepackage{caption}
\usepackage{algpseudocode}
\usepackage[ruled,vlined,linesnumbered]{algorithm2e}
\newenvironment{proof-sketch}{\noindent{\emph{Proof (Sketch).}}}{\bigskip}

\newcommand{\myi}{(\emph{i})\xspace}
\newcommand{\myii}{(\emph{ii})\xspace}
\newcommand{\myiii}{(\emph{iii})\xspace}
\newcommand{\myiv}{(\emph{iv})\xspace}

\newcommand{\A}{\mathcal{A}} \newcommand{\B}{\mathcal{B}}
\newcommand{\C}{\mathcal{C}}

\newcommand{\I}{\mathcal{I}} 
 
\newcommand{\M}{\mathcal{M}} 
\renewcommand{\O}{\mathcal{O}} 
 
\renewcommand{\S}{\mathcal{S}} \newcommand{\T}{\mathcal{T}}

\newcommand{\D}{\mathcal{D}}

\newcommand{\Lora}{\Longrightarrow}

\newcommand{\Lola}{\Longleftarrow}

\newcommand{\tup}[1]{\langle #1\rangle}            

\newcommand{\dlliteaid}{\textit{DL-Lite}_{\A ,id}\xspace}

\newcommand{\PTIME}{\textsc{PTime}\xspace}
\newcommand{\EXPTIME}{\textsc{ExpTime}\xspace}
\newcommand{\NP}{\textrm{NP}\xspace}

\newcommand{\MOD}{\textit{MOD}}

\newcommand{\cert}{\mathit{cert}}

\newcommand{\sparql}{\texttt{SPARQL}\xspace}

\renewcommand{\D}{\texttt{D}\xspace}

\renewcommand{\T}{\texttt{T}\xspace}

\newcommand{\sem}{\mathit{sem}}

\newcommand{\dom}{\mathit{dom}}

\newcommand{\nulls}{\mathit{null}}
\newcommand{\tuplehead}{\mathit{tup}}
\newcommand{\sat}{\mathit{sat}}
\newcommand{\pr}{\mathit{pr}}

\title{Preliminary results on Ontology-based Open Data Publishing}

\author{Gianluca Cima}

\institute{
  Dipartimento di Ingegneria Informatica, Automatica e Gestionale\\
  Sapienza Universit\`a di Roma\\
  \href{mailto:cima@diag.uniroma1.it}{cima@diag.uniroma1.it}
  }

\begin{document}
\nocite{*}
\maketitle
\begin{abstract}
  Despite the current interest in Open Data publishing, a formal and
  comprehensive methodology supporting an organization in deciding
  which data to publish and carrying out precise procedures for
  publishing high-quality data, is still missing. In this paper we
  argue that the Ontology-based Data Management paradigm can provide a
  formal basis for a principled approach to publish high-quality,
  semantically annotated Open Data. We describe two main approaches to
  using an ontology for this endeavor, and then we present some
  technical results on one of the approaches, called bottom-up, where
  the specification of the data to be published is given in terms of
  the sources, and specific techniques allow deriving suitable
  annotations for interpreting the published data under the light of
  the ontology.
\end{abstract}
\section{Introduction} \label{sec:introduction}

In many aspects of our society there is growing awareness and consent
on the need for data-driven approaches that are resilient, transparent
and fully accountable. But to achieve a data-driven society, it is
necessary that the data needed for public goods are readily
available. Thus, it is no surprising that in recent years, both public
and private organizations have been faced with the issue of publishing
Open Data, in particular with the goal of providing data consumers
with suitable information to capture the semantics of the data they
publish.  Significant efforts have been devoted to defining guidelines
concerning the management and publication of Open Data. Notably, the
W3C\footnote{World Wide Web Consortium: \url{https://www.w3.org/}} has
formed a working group, whose objective is the release of a first
draft on Open Data Standards\footnote{Data on the Web Best Practice:
  \url{https://www.w3.org/TR/dwbp/}}. The focus of the document are
areas such as metadata, data formats, data licenses, data quality,
etc., which are treated in very general terms, with no reference to
any specifical technical methodology. More generally, although there
are several works on platforms and architectures for publishing Open
Data, there is still no formal and comprehensive methodology
supporting an organization in \myi deciding which data to publish, and
\myii carrying out precise procedures for publishing and documenting
high-quality data. One of the reasons of this lack of formal methods
is that the problem of Open Data Publishing is strictly related to the
problem of managing the data within an organization. Indeed, a
necessary prerequisite for an organization for publishing relevant and
meaningful data is to be able to manage, maintain and document its own
information system. The recent paradigm of Ontology-based Data
Management (OBDM)~\cite{Lenz11} (used and experimented in practice in
the last years, see, e.g.,~\cite{ACCG*14}) is an attempt to provide
the principles and the techniques for addressing this challenge. An
OBDM system is constituted by an ontology, the data sources forming
the information system, and the mapping between the ontology and the
sources. The ontology is a formal representation of the domain
underlying the information system, and the mapping is a precise
specification of the relationship between the data at the sources and
the concepts in the ontology.

In this paper we argue that the OBDM paradigm can provide a formal basis for
a principled approach to publish high-quality, semantically annotated Open Data. The 
most basic task in Open Data is the extraction of the \emph{correct content} for the dataset(s) to 
be published, where by ``content'' we mean both the extensional information (i.e., facts 
about the domain of interest) conveyed by the dataset, and the intensional knowledge 
relevant to document such facts (e.g., concepts that intensionally describe facts), and 
``correct'' means that the aspect of the domain captured by the dataset is coherent with 
a requirement formally expressed in the organization.

Current practices for publishing Open Data focus essentially on
providing extensional information (often in very simple forms, such as
CSV files), and they carry out the task of documenting data mostly by
using metadata expressed in natural languages, or in terms of record
structures. As a consequence, the semantics of datasets is not
formally expressed in a machine-readable form. Conversely, OBDM opens
up the possibility of a new way of publishing data, with the idea of
annotating data items with the ontology elements that describe them in
terms of the concepts in the domain of the organization. When an OBDM
is available in an organization, an obvious way to proceed to Open
Data publication is as follows: \myi express the dataset to be
published in terms of a \sparql query over the ontology, \myii compute
the certain answers to the query, and \myiii publish the result of the
certain answer computation, using the query expression and the
ontology as a basis for annotating the dataset with suitable metadata
expressing its semantics. We call such method \emph{top-down}. Using
this method, the ontology is the heart of the task: it is used for
expressing the content of the dataset to be published (in terms of a
query), and it is used, together with the query, for annotating the
published data.

Unfortunately, in many organizations (for example, in Public Administrations) it may be the case that people are not ready yet to manage their information systems through  the OBDM paradigm. In these cases, the bottom-up approach could be more appropriate.
For example, in the Italian Public Administration system, it is very unlikely that local administration people are able to express their queries over the ontology using \sparql. Typically, the ontology and the mapping have been designed by third parties, with no or little involvement with IT people responsible of the local administration information system.  In other words, these people probably cannot follow the top-down approach, and they are more confident to express the specification of the dataset to be published directly in terms of the source structures (i.e., the relational tables in their databases), or, more generally, in terms of a view over the sources. But how can we automatically publish both the content and the semantics of the
dataset if its specification is given in terms of the data sources? We
argue that we can achieve this goal by following what we call the
\emph{bottom-up} approach: the organization expresses its publishing
requirement as a query over the sources, and, by using the ontology
and the mapping, a suitable algorithm computes the corresponding query
over the ontology. With such query at hand, we have reduced the
problem in such a way that the top-down approach can now be followed,
and the required data can be published according to the method
described above. So, at the heart of the bottom-up approach there is a
conceptual issue to address:

\begin{quote}
  \emph{"Given a query $Q$ over the sources, which is the query over
    the ontology that characterizes $Q$ at best (independently from
    the current source database)?"}
\end{quote}
Note that the answer to this question is relevant also for other tasks
related to the management of the information system, e.g., the task of
explaining the semantics of the various data sources within the
organization. The question implicitly refers to a sort of reverse
engineering problem, which is a novel aspect in the investigation of
both OBDM and data integration. Indeed, most of (if not all) the
literature about managing data sources through an ontology (see,
e.g.,~\cite{PLCD*08,CDLL*09}), or more generally, about data
integration~\cite{Lenz02} assume that the user query is expressed over
the global schema, and the goal is to find a \emph{rewriting} (i.e., a
query over the source schema) that captures the original query in the
best way, independently from the current source database. Here, the
problem is reversed, because we start with a source query and we aim
at deriving a corresponding query over the ontology, called a
source-to-target rewriting.

In this paper we study the above described bottom-up approach, and
provide the following contributions.
\begin{itemize}
\item We introduce the concept of \emph{source-to-target rewriting}
  (see Section \ref{sec:rewriting}), the main technical notion underlying the bottom-up
  approach, and we describe two computation problems related to it,
  namely the recognition problem, and the finding problem. The former
  aims at checking whether a query over the ontology is a
  source-to-target rewriting of a given query over the sources, taking
  into account the mapping between the sources and the ontology. The
  latter aims at computing a suitable source-to-target rewriting of a
  given source query, with respect to the mapping.
\item We discuss two different semantics for source-to-target
  rewritings, one based on the logical models of the OBDM
  specification, and one based on certain answers. The former is
  somehow the natural choice, given the first-order semantics behind
  OBDM. The latter is a significant alternative, that may better
  capture the intuition of a user who is accustomed to think of query
  semantics in terms of certain answers.
\item We show that, although the ideal notion is the one of ``exact''
  source-to-target rewriting, it is important to resort to
  approximations to exact rewriting when exactness cannot be
  achieved. For this reason, we introduce the notion of sound and
  complete source-to-target rewritings.
\item For the case of complete source-to-target rewritings, we present
  algorithms both for the recognition (Section \ref{sec:RecognitionProblem}), and for the finding
  (Section \ref{sec:finding}) problem, in particular for the setting where the
  ontology is expressed in $\dlliteaid$, and the queries involved in
  the specification are conjunctive queries.
\end{itemize}

\section{Preliminaries}\label{sec:preliminaries}

We assume familiarity with classical databases~\cite{AbHV95}, Description Logics~\cite{BCMNP03}, and the OBDM paradigm. In this section, we \myi review the most basic notions of non-ground instances, and their correlation with conjunctive queries; \myii briefly discuss the chase of a possible non-ground instance; \myiii discuss the relevant aspects of notation we use in the following regarding the OBDM paradigm.

For a possible non-ground instance $\D$, we assume that each value in $\dom(\D)$, i.e., the set of values occurring in $\D$, comes from the union of two fixed disjoint infinite sets: the set \emph{Const} of all \emph{constants}, and the set $\mathit{Null_{\D}}$ of all \emph{labeled nulls}. We also let $\nulls(\D):=\dom(\D)\cap\mathit{Null_{\D}}$. In particular, each labeled null in a non-ground instance is treated as an unknown value (and hence, an \emph{incomplete information}), rather than to a non-existent value~\cite{Zani82}. Thus, a non-ground instance represents a number of ground instances obtained by assigning constants to each labeled null. More precisely, let $\D$ be a non-ground instance, and $v$ be a mapping $v:\emph{null}(\D)\rightarrow\emph{Const}$. Then, $v$ is called a \emph{valuation} of $\D$, and we indicate, with $v(\D)$, the ground instance obtained from $\D$ by replacing elsewhere each labeled null $x\in\D$ with $v(x)$. We also extend this to tuples, that is, given a tuple $\overline{u}=(u_1,...,u_n)$ of both constants and labeled nulls, with $v(\overline{u})$ we indicate the tuple $(v'(u_1),...,v'(u_n))$, where $v'(u_i)=u_i$ if $u_i$ is a constant; otherwise ($u_i$ is a labeled null), $v'(u_i)=v(u_i)$. Given an instance $\D$ it is possible to construct in linear time a boolean CQ $q_{\D}$ that fully captures it, and vice versa. We also let $q_{\D}(\overline{x})$ denoting the transformation of $q_{\D}$ by removing the existential quantification of the variables $\overline{x}$ in $q_{\D}$. Moreover, given a non-boolean CQ $q$ (with $\overline{x}$ as distinguished variables), we associate to it the instance $\D_{q}$ by considering the variables in $\overline{x}$ as if they were existentially quantified. For ease of presentation, we extend CQs to allow also queries of the form $\{\overline{x}\ |\ \bot(\overline{x})\}$ and $\{\overline{x}\ |\ \top(\overline{x})\}$, with their usual meaning. We also denote with $\tuplehead(q)$ the tuple composed by the terms in \emph{head} of $q$.

Given a source schema $\S$; a target schema $T$; a set $\M$ of \emph{st-tgds} (i.e., assertions of the form $\forall\overline{x},\overline{y}(\phi(\overline{x},\overline{y})\rightarrow\exists\overline{z}\varphi(\overline{x},\overline{z}))$, where $\phi$ is a CQ over $\S$, and $\varphi$ is a CQ over $\T$); and a set $\Sigma_{t}$ of \emph{egds} (i.e., assertions of the form $\forall\overline{x}(\phi_{\T}(\overline{x})\rightarrow(x_1=x_2))$, where $\phi_{\T}$ is a CQ over $\T$, and $x_1,x_2$ are among the variables in $\overline{x}$), the \emph{chase} procedure of a possibly non-ground source instance $\D$ consists in: \myi the chase of $\D$ w.r.t. $\M$, where, for every st-tgd $\phi(\overline{x},\overline{y})\rightarrow\exists\overline{z}\varphi(\overline{x},\overline{z})$ in $\M$ and for every pair of tuples $(\overline{a},\overline{b})$ such that $\D\models\phi(\overline{a},\overline{b})$, there is the introduction of new facts in the instance $J$ of the target schema $T$ so that $\varphi(\overline{a},\overline{u})$ holds, where $\overline{u}$ consists in a fresh tuple of distinct labeled nulls coming from an infinite set $\mathit{Null_{\Sigma}}$ disjoint from $\mathit{Null_{\D}}$; \myii the chase of $J$ w.r.t. $\Sigma_{t}$, where, for every egd $\forall\overline{x}(\phi_{\T}(\overline{x})\rightarrow(x_1=x_2))$ and for every tuple $\overline{a}$ such that $J\models\phi_{\T}(\overline{a})$ and $a_1\neq a_2$, we equate the two terms. Equating $a_1$ with $a_2$ means choosing one of the two so that the other is replaced elsewhere in $J$ by the one chosen. In particular, if one is a labeled null and the other is a constant, then the chase choose the constant; if both are labeled nulls, one coming from $\mathit{Null_{\D}}$ and the other from $\mathit{Null_{\Sigma}}$, it always choose the one coming from $\mathit{Null_{\D}}$; if both are constants, then the chase fail. Moreover, with $\psi$ we denote the set of equalities applied by the chase of $J$ w.r.t. a set of egds on variables coming from $\mathit{Null_{\D}}$. This can be done by keeping track of the substitution applied by the chase. For example, if the chase equates the variable $y\in\mathit{Null_{\D}}$ with the variable $x\in\mathit{Null_{\D}}$, and then equates the variable $z\in\mathit{Null_{\D}}$ with the variable $w\in\mathit{Null_{\D}}$, and then $w$ with the constant $a$, given the tuple $(x,y,z,w)$, $\psi(x,y,z,w)$ indicates the tuple $(x,x,a,a)$. Note that, we can compute the certain answers of a boolean union of CQs (UCQ) $q$ with at most one inequality per disjunct by splitting $q$ as a boolean UCQ $q^{1\neq}$ with exactly one inequality per disjunct, and a boolean UCQ $q^{0\neq}$ with no inequality per disjunct. The key idea is that the negation of $q^{1\neq}$ consists in a set of egds, hence, the certain answers of $q$ can be computed by applying the chase procedure over the instance $J$ (i.e., the instance produced by the chase of $\C$ w.r.t. $\M$ and $\Sigma_{t}$) w.r.t. $\neg q^{1\neq}$, where, if the chase fail then the answer is $\mathit{true}$; otherwise, if the instance $J'$ produced satisfy one of the conjunctive query in $q^{0\neq}$, then the answer is $\mathit{true}$, else the answer is $\mathit{false}$. We refer to~\cite{FKMP03} for more details.\\
Given an OBDM specification $\I=\tup{\O,\M,\S}$, where $\O$ is a TBox, and $\M$ is a set of st-tgds, and given a non-ground source instance $\D$ for $\S$, and a set of egds $\Sigma_{t}$, we denote with $\A_{\D,\Sigma}$, where $\Sigma=\M\cup\Sigma_{t}$, the ABox computed as follows: \myi chase the non-ground source instance $\D$ w.r.t. $\Sigma$; \myii \emph{freeze} the instance (or equivalently, the ABox with variables) obtained, i.e., variables in this instance are now considered as constant. Note that, such ABox $\A_{\D,\Sigma}$ may also not exists due to the failing of the chase, in this case, we denote $\A_{\D,\Sigma}$ with the symbol $\bot$.

For an OBDM specification $\I=\tup{\O,\M,\S}$, and for a source database $\C$ for $\I$ (i.e., a ground instance over the schema $\S$), we denote by $\sem^{\C}(\I)$ the set of \emph{models} $\B$ for $\I$ relative to $\C$ such that: \myi $\B\models\O$; \myii $(\B,\C)\models\M$. Given a query $q$ over $\O$, we denote by $\cert(q,\I,\C)$ the set of \emph{certain answers} to $q$ in $\I$ relative to $\C$. It is defined as: $\cert(q,\I,\C)=\bigcap\{q^{\B}\ |\ \B\in\sem^{\C}(\I)\}$ if $\sem^{\C}(\I)\neq\emptyset$; otherwise, $\cert(q,\I,\C)=\mathit{AllTup}(q,\C)$, where $\mathit{AllTup}(q,\C)$ is the set of all possible tuples of constants in $\C$ whose arity is the one of the query $q$. Furthermore, given a $\dlliteaid$~\cite{CDLL*09} TBox $\O$ and a $\dlliteaid$ ABox $\A$ we are able to: \myi check whether $\tup{\O,\A}$ is satisfiable by computing the answers of a suitable boolean query $Q_{\sat}$ (a UCQ with at most one inequality per disjunct) over the ABox $\A$ considered as a relational database. We see $Q_{\sat}$ as the union of $Q_{\sat}^{0\neq}$ (the UCQ containing every disjunct not comprising inequalities in $Q_{\sat}$) and $Q_{\sat}^{1\neq}$ (the UCQ containing every disjunct comprising inequalities in $Q_{\sat}$); \myii compute the \emph{certain answers} to a UCQ $Q_{g}$ over a satisfiable $\tup{\O,\A}$, denoted with $\cert(Q_{g},\tup{\O,\A})$, by producing a \emph{perfect reformulation} (denoted as a function $\pr(\cdot)$) of such query, and then computing the answers of $\pr(Q_{g})$ over the ABox $\A$ considered as a relational database. See~\cite{CDLLR07} for more details.

\section{The notion of source-to-target rewriting}\label{sec:rewriting}

In what follows, we implicitly refer to \myi an OBDM specification
$\I=\tup{\O,\M,\S}$; \myii a query $Q_{s}$ over the source schema
$\S$; \myiii a query $Q_{g}$ over the ontology $\O$.  

As we said in the introduction, there are at least two different ways to
formally define a source-to-target rewriting (\emph{s-to-t rewriting}
in the following) for each of the three variants, namely ``exact'',
``complete'', and ``sound''. The first one is captured by the
following definition.

\begin{definition}
  \normalfont $Q_{g}$ is a \emph{complete} (resp., \emph{sound},
  \emph{exact}) \emph{s-to-t rewriting} of $Q_{s}$ with respect to
  $\I$ under the \emph{model-based semantics}, if for each source
  database $\C$ and for each model $\B\in\sem^{\C}(\I)$, we have that
  $Q_{s}^{\C}\subseteq Q_{g}^{\B}$ (resp., $Q_{g}^{\B}\subseteq
  Q_{s}^{\C}$, $Q_{g}^{\B}=Q_{s}^{\C}$).
\end{definition}

Intuitively, a complete s-to-t rewriting of $Q_{s}$ w.r.t. $\I$ under
the model-based semantics is a query over $\O$ that, when evaluated
over a model $\B\in\sem^{\C}(\I)$ for a source database $\C$, returns
all the answers of the evaluation of $Q_{s}$ over $\C$. In other
words, for every source database $\C$, the query $Q_{g}$ over $\O$
captures all the semantics that $Q_{s}$ expresses over $\C$. Similar
arguments hold for the notions of sound and exact s-to-t rewriting
under this semantics. Moreover, from the formal definition of
source-to-target rewriting and the usual definition of
target-to-source rewriting (simply called rewriting) used in data
integration, it is easy to see that $Q_{g}$ is a complete (resp.,
sound) source-to-target rewriting of $Q_{s}$ w.r.t. $\I$ under the
model-based semantics, if and only if $Q_{s}$ is a sound (resp.,
complete) rewriting of $Q_{s}$ w.r.t. $\I$, implying that, $Q_{g}$ is
an exact source-to-target rewriting of $Q_{s}$ w.r.t. $\I$ under the
model-based semantics, if and only if $Q_{s}$ is an exact rewriting of
$Q_{g}$ w.r.t. $\I$.

The second possible way to formally define a source-to-target
rewriting is as follows.
\begin{definition}
  \normalfont $Q_{g}$ is a \emph{complete} (resp., \emph{sound},
  \emph{exact}) \emph{s-to-t rewriting} of $Q_{s}$ with respect to
  $\I$ under the \emph{certain answers-based semantics}, if for each
  source database $\C$ such that $\sem^{\C}(\I)\neq\emptyset$, we have
  that $Q_{s}^{\C}\subseteq \cert(Q_{g},\I,\C)$ (resp.,
  $\cert(Q_{g},\I,\C)\subseteq Q_{s}^{\C}$,
  $Q_{s}^{\C}=\cert(Q_{g},\I,\C)$).
\end{definition}
In this new semantics, in order to capture a query $Q_{s}$ over $\S$,
we resort to the notion of certain answers. Indeed, a complete s-to-t
rewriting of $Q_{s}$ w.r.t. $\I$ under the certain answers-based
semantics is a query over $\O$ such that, when we compute its certain
answers for a source database $\C$, we get all the answers of the
evaluation of $Q_{s}$ over $\C$. As before, similar arguments hold
for the notions of sound and exact s-to-t rewriting under this
semantics. Note also the strong correspondence between the exact
s-to-t rewriting under the certain answers-based semantics and the notion of \emph{perfect rewriting}. We remind that a perfect rewriting of $Q_{g}$
w.r.t. $\I$ is a query $Q_{s}$ over $\S$ that computes
$\cert(Q_{g},\I,\C)$ for every source database $\C$ such that
$\sem^{\C}(\I)\neq\emptyset$~\cite{CDLV05}. Indeed, we have that
$Q_{g}$ is an exact s-to-t rewriting of $Q_{s}$ w.r.t. $\I$ under the
certain answers-based semantics if and only if $Q_{s}$ is a perfect
rewriting of $Q_{g}$ w.r.t. $\I$.  Note that the above observations
imply that the two semantics are indeed different, since it is
well-known that the two notions of exact rewriting and perfect
rewriting of $Q_{g}$ w.r.t. $\I$ are different. The difference between
the two semantics is confirmed by the following example.
\begin{example}
  $\O:=\emptyset$ (i.e., no TBox assertions in $\O$); $\S$ contains a
  binary relation $r_{1}$ and a unary relation $r_{2}$;
  $\M:=\{\forall x\forall y(r_{1}(x,y)\rightarrow\textsf{G}(x,y)), \forall x(r_{2}(x)\rightarrow\exists
  Y.\textsf{G}(x,Y))\}$; $Q_{s}:=\left\{(w)\ |\ \exists
    Z.r_{1}(Z,w)\right\}$; $Q_{g}:=\left\{(w)\ |\ \exists
    Z.\textsf{G}(Z,w)\right\}$.

  It is easy to see that $Q_{g}$ is a sound s-to-t rewriting of
  $Q_{s}$ w.r.t. $\I$ under the certain answers-based semantics (more
  precisely, it is an exact s-to-t rewriting of $Q_{s}$ w.r.t. $\I$
  under such semantics), while it is not sound under the model-based
  semantics. In fact, for the source database $\C$ with
  $r_{1}^{\mathcal{C}}=\{(a,b)\}$ and $r_{2}^{\mathcal{C}}=\{(c)\}$,
  and for the model $\mathcal{B}$ with
  $\textsf{G}^{\mathcal{B}}=\{(a,b),(c,d)\}$, we have
  $\B\in\sem^{\C}(\I)$, and $Q_{g}^{\B}\not\subseteq Q_{s}^{\C}$. \qed
\end{example}
Intuitively, for the sound case, the model-based semantics is too
strong, in the sense that under such semantics, a model $\B$ may
contain not only facts depending on how data in the source $\C$ are
linked to $\O$ through $\M$, but additionally arbitrary facts, with
the only constraint of satisfying $\O$. One might think that, in
order to address this issue, it is sufficient to resort to a sort of
minimizations of the models of $\O$. Actually, the above example shows
that, even if we restrict the set of models to the set of minimal
models (i.e., models $\B$ such that \myi $\B\in\sem^{\C}(\I)$ and
\myii there is no model $\B'\in\sem^{\C}(\I)$ such that
$\B'\subset\B$), and adopt a semantics like the model-based one
but restricted to the set of minimal models, $Q_{g}$ is still not a
sound s-to-t rewriting (this can be seen considering that the target
database $\B$ defined earlier is a minimal model).

Observe that the above considerations show the difference in the two
semantics by referring to sound and exact s-to-t rewritings. It is
interesting to ask whether the difference shows up when restricting
our attention to complete rewritings. The following proposition deals
with this question.

\begin{proposition}
  $Q_{g}$ is a complete s-to-t rewriting of $Q_{s}$ with respect to
  $\I$ under the model-based semantics if and only if it is so under
  the certain answers-based semantics.
\end{proposition}

\begin{proof-sketch}
One direction is trivial. Indeed, when $Q_{g}$ is a complete s-to-t rewriting of $Q_{s}$ with respect to $\I$ under the model-based semantics, by definition of certain answers, for each source database $\C$ such that $\sem^{\C}(\I)\neq\emptyset$ we have that $Q_{s}^{\C}\subseteq \cert(Q_{g},\I,\C)$. For the other direction, suppose that $Q_{g}$ is not a complete s-to-t rewriting of $Q_{s}$ w.r.t. $\I$ under the model-based semantics. It follows that, there exists a source database $\C$ and a model $\B\in\sem^{\C}(\I)$ such that $Q_{s}^{\C}\not\subseteq Q_{g}^{\B}$, implying that, $Q_{s}^{\C}\not\subseteq \cert(Q_{g},\I,\C)$, which, in turn, implies that $Q_{g}$ is not a complete s-to-t rewriting of $Q_{s}$ w.r.t. $\I$ under the certain answers-based semantics. \qed
\end{proof-sketch}

Obviously, the query over the ontology which captures at best a given
query $q$ over the source schema is the exact s-to-t rewriting of
$q$. However, the following example shows that even for very simple
OBDM specifications, an exact s-to-t rewriting of even trivial queries,
may not exist.
\begin{example} \label{ex:NoExact} $\O:=\emptyset$ (i.e., no TBox
  assertions in $\O$); $\S$ contains two unary relations
  $\mathit{man}$ and $\mathit{woman}$;
  $\M:=\{\forall x(\mathit{man}(x)\rightarrow\textsf{Person}(x)),
  \forall x(\mathit{woman}(x)\rightarrow\textsf{Person}(x))\}$; $Q_{s}:=\left\{(x)\
    |\ \mathit{woman}(x)\right\}$.

  It is possible to show that the only sound s-to-t rewriting of
  $Q_{s}$ w.r.t. $\I$ under both semantics is the query
  $Q_{g}:=\left\{ (x)\ |\ \bot(x)\right\}$, which is obviously not a
  complete s-to-t rewriting of $Q_{s}$ w.r.t. $\I$ neither under the
  model-based semantics, nor under the certain answers-based
  semantics. On the other hand, the most immediate and intuitive
  complete s-to-t rewriting of $Q_{s}$ w.r.t. $\I$ is the query
  $Q'_{g}:=\left\{ (x)\ |\ \textsf{Person}(x)\right\}$. Furthermore,
  as we will see in Section \ref{sec:finding}, this query is an
  ``optimal'' complete s-to-t rewriting of $Q_{s}$ w.r.t. $\I$, where
  the term \emph{optimal} will be precisely defined. \qed 
\end{example}

As we said in the introduction, in the rest of this paper we focus on
complete s-to-t-rewritings. In particular, we will address both the
recognition problem (see Section \ref{sec:RecognitionProblem}), and
the finding problem (see Section \ref{sec:finding}) in a specific
setting, characterized as follows:
\begin{itemize}
\item The ontology $\O$ in an OBDM specification $\I=\tup{\O,\M,\S}$
  is expressed as a TBox in $\dlliteaid$.
\item The mapping $\M$ in $\I$ is a set of GLAV mapping assertions (or,
  \emph{st-tgds}), where each assertion expresses a correspondence
  between a conjunctive query over the source schema and a conjunctive
  query over the ontology.
\item In the recognition problem, both the query over the source
  schema and the query over the ontology are conjunctive
  queries. Similarly, in the finding problem, the query over the
  source schema is a conjunctive query.
\end{itemize}

\section{The recognition problem for complete s-to-t
  rewritings} \label{sec:RecognitionProblem}

We implicitly refer to the setting described at the end of the
previous section. The \emph{recognition problem} associated to the
complete s-to-t rewriting is the following decision problem: Given an
OBDM specification $\I=\tup{\O,\M,\S}$, a query $Q_{s}$ over the
source schema $\S$, and a query $Q_{g}$ over the ontology $\O$, check
whether $Q_{g}$ is a complete s-to-t rewriting of $Q_{s}$ with respect to
$\I$. The next lemma is the starting point of our solution. 
\begin{lemma} \label{lm:NotComplete} $Q_{g}$ is not a complete s-to-t
  rewriting of $Q_{s}$ with respect to $\I=\tup{\O,\M,\S}$ if and only
  if there is a valuation $v$ of $D_{Q_{s}}$ and a model
  $\B\in\sem^{v(\D_{Q_{s}})}(\I)$ such that
  $v(\tuplehead(Q_{s}))\not\in Q_{g}^{\B}$.
\end{lemma}
\begin{proof}
  "$\Lola$" Suppose that there exists a valuation $v$ of $\D_{Q_{s}}$
  and a model $\B\in\sem^{v(\D_{Q_{s}})}(\I)$ such that
  $v(\tuplehead(Q_{s}))\not\in Q_{g}^{\B}$. Obviously,
  $v(\tuplehead(Q_{s}))\in Q_{s}^{v(\D_{Q_{s}})}$. It follows that,
  there exist a source database $v(\D_{Q_{s}})$, a model
  $\B\in\sem^{v(\D_{Q_{s}})}(\I)$, and a tuple $v(\tuplehead(Q_{s}))$
  such that $v(\tuplehead(Q_{s}))\not\in Q_{g}^{\B}$ and
  $v(\tuplehead(Q_{s}))\in Q_{s}^{v(\D_{Q_{s}})}$.

  "$\Lora$" Suppose that $Q_{g}$ is not a complete s-to-t rewriting of
  $Q_{s}$ w.r.t. $\I$, i.e., there is a source database $\C$, a model
  $\B\in\sem^{\C}(\I)$, and a tuple $t$ such that $t\in Q_{s}^{\C}$
  and $t\not\in Q_{g}^{\B}$. The fact that $t\in Q_{s}^{\C}$ implies
  the existence of a homomorphism $v:\D_{Q_{s}}\rightarrow\C$ such
  that $v(\tuplehead(Q_{s}))=t$. Note also, that since $\C$ is a
  \emph{ground} instance, $v$ is a valuation of $\D_{Q_{s}}$ such that
  $v(\D_{Q_{s}})\subseteq\C$. Obviously,
  $\B\in\sem^{v(D_{Q_{s}})}(\I)$, this can be seen by considering that
  \myi $\B\models\O$ is true from the supposition that
  $\B\in\sem^{\C}(\I)$; and \myii $(\B,v(\D_{Q_{s}}))\models\M$ is
  true by considering that, $(\B,\C)\models\M$ (which holds from the
  supposition that $\B\in\sem^{\C}(\I)$), $v(D_{Q_{s}})\subseteq\C$,
  and the queries in $\M$ are \emph{monotone} queries. It follows
  that, there is a valuation $v$ of $\D_{Q_{s}}$ and a model
  $\B\in\sem^{v(\D_{Q_{s}})}(\I)$ such that
  $v(\tuplehead(Q_{s}))\not\in Q_{g}^{\B}$. \qed
\end{proof}
Relying on the above lemma, we are now ready to present the algorithm
\textsf{CheckComplete} for the recognition problem. 

\begin{algorithm}[H]
  \caption{CheckComplete($\I$, $Q_{s}$, $Q_{g}$)}\label{alg:CheckComplete}
  \begin{algorithmic}[1]
    \Require{OBDM specification $\I=\tup{\O,\S,\M}$, query $Q_{s}$
      over $\S$, query $Q_{g}$ over $\O$.}  \Ensure{\emph{true} or \emph{false}.}

    \BlankLine

    \State Compute $\D_{Q_{s}}$ from $Q_{s}$ (i.e., the instance,
    possibly with incomplete information, associated to the query
    $Q_{s}$), and denote it with $\D$.

    \State Compute $\A_{\D,\Sigma}$, where $\Sigma=\M\cup\neg
    Q_{\sat}^{1\neq}$, and let $\psi$ be the set of equality applied to the
    variables in $\D$ by the chase.

    \State If $\A_{\D,\Sigma}=\bot$, then return \emph{true}.

    \State If the evaluation of $Q_{\sat}^{0\neq}$ over $\A_{\D,\Sigma}$ considered as a relational database is $\{()\}$ (i.e., $\A_{\D,\Sigma}\models Q_{\sat}^{0\neq}$), then return
    \emph{true}.

    \State If
    $\psi(\tuplehead(Q_{s}))\in\cert(Q_{g},\tup{\O,\A_{\D,\Sigma}})$
    then return \emph{true}, else return \emph{false}.
  \end{algorithmic}
\end{algorithm}
The next theorem establishes the correctness of the above algorithm.
\begin{theorem} \label{thm:CheckComplete} \textsf{CheckComplete}($\I$,
  $Q_{s}$, $Q_{g}$) terminates, and returns \emph{true} if and only if
  $Q_g$ is a complete s-to-t rewriting of $Q_s$ w.r.t. $\I$.
\end{theorem}
\begin{proof-sketch}
Termination of the algorithm easily follows by the termination of the chase procedure, and by the obvious termination of computing the certain answers of a CQ over $\tup{\O,\A_{\D,\Sigma}}$.

For the "$\Lora$" direction, suppose that the algorithm returns false, i.e., $\A_{\D,\Sigma}\not\models Q_{\sat}^{0\neq}$, and $\psi(\tuplehead(Q_{s}))\not\in\cert(Q_{g},\tup{\O,\A_{\D,\Sigma}})$. Now, if we extend $\psi(\D)$ by considering the freezing of this instance (i.e., variables are now considered as constants), it is easy to see that we obtain a valuation $v$ of $\D$ such that $(v(\D),\A_{\D,\Sigma})\models\M$, and such that $\sem^{v(\D)}(\I)\neq\emptyset$. Moreover, the fact that $\psi(\tuplehead(Q_{s}))\not\in\cert(Q_{g},\tup{\O,\A_{\D,\Sigma}})$, implies, by the property of certain answers, that there is at least one model $\B\models\tup{\O,\A_{\D,\Sigma}}$, and hence $\B\in\sem^{v(\D)}(\I)$ (because $(v(\D),\A_{\D,\Sigma})\models\M$) such that $v(\tuplehead(Q_{s}))\not\in Q_{g}^{\B}$. It follows, from Lemma \ref{lm:NotComplete}, that $Q_{g}$ is not a complete s-to-t rewriting of $Q_{s}$ w.r.t. $\I$.

For the "$\Lola$" direction, in the cases that $\A_{\D,\Sigma}=\bot$ or $\A_{\D,\Sigma}\models Q_{\sat}^{0\neq}$, it is easy to see that for every valuation $v$ of $\D$, either the chase of $v(\D)$ will fail, or every ABox $\A$ such that $(v(\D),\A)\models\M$ and $\A\models\neg Q_{\sat}^{1\neq}$ will be such that $\A\models Q_{\sat}^{0\neq}$, implying that, for every valuation $v$ of $\D$, $\sem^{v(\D)}=\emptyset$. It follows, from Lemma \ref{lm:NotComplete}, that in this case $Q_{g}$ is a complete s-to-t rewriting of $Q_{s}$ w.r.t. $\I$. While, in the cases that $\psi(\tuplehead(Q_{s}))\in\cert(Q_{g},\tup{\O,\A_{\D,\Sigma}})$, it is easy to see that, for every valuation $v$ of $\D$ either $\sem^{v(\D)}=\emptyset$, or if we compute $\A_{v(\D),\Sigma}$, we have that $v(\tuplehead(Q_{s}))\in\cert(Q_{g},\tup{\O,\A_{v(\D),\Sigma}})$. More generally, every $\A$ obtained by chasing $v(\D)$ w.r.t. $\M$ and $\neg Q_{\sat}^{1\neq}$, and then choosing arbitrary constants for the possible remaining variables, is such that $v(\tuplehead(Q_{s}))\in\cert(Q_{g},\tup{\O,\A})$. Hence, for every model $\B$ such that $\B\models\tup{\O,\A}$, we have that $v(\tuplehead(Q_{s}))\in Q_{g}^{\B}$. Also, we observe that the set of models $\sem^{v(\D)}$ coincides with the set of all models $\B$ such that $\B\models\tup{\O,\A}$ for all the possible ABox $\A$ obtained using the above procedure. It follows that, for every possible valuation $v$ of $\D$ and for every possible $\B\in\sem^{v(\D)}(\I)$, we have that $v(\tuplehead(Q_{s}))\in Q_{g}^{\B}$, implying, from Lemma \ref{lm:NotComplete}, that also in this case $Q_{g}$ is a complete s-to-t rewriting of $Q_{s}$ w.r.t. $\I$. \qed
\end{proof-sketch}
As for complexity issues of the algorithm, we observe: \myi it runs in $\PTIME$ in the size of $Q_{s}$. Indeed, computing $\D$ (the instance associated to the query $Q_{s}$) can be done in linear time, and chasing an instance in the presence of a weakly acyclic set of tgds (as in our case) is $\PTIME$ in the size of $\D$ ($\M$ and $\Sigma$ are considered fixed); \myii it runs in $\PTIME$ in the size of $\O$. Indeed, $Q_{\sat}$ and the evaluation of the certain answers of $Q_{g}$ can be both computed in $\PTIME$ in the size of $\O$; \myiii it runs in $\EXPTIME$ in the size of $\M$. This can be seen from the obvious $\EXPTIME$ process of transferring data from $\D$ to $\A_{\D,\Sigma}$; \myiv the problem is $\NP$-complete in the size of $Q_{g}$ because computing the certain answers of a UCQ query is $\NP$-complete in the size of the query (\emph{query complexity}).

\section{Finding optimal complete s-to-t
  rewritings}\label{sec:finding}

In this section we study the problem of finding optimal complete
s-to-t rewritings. The first question to ask is which rewriting we
chose in the case where several complete rewritings exist. The obvious
choice is to define the notion of ``optimal'' complete s-to-t
rewriting: one such rewriting $r$ is optimal if there is no complete
s-to-t rewriting that is contained in $r$. In order to formalize this
notion, we introduce the following definitions (where $\MOD(\O)$ denotes the set of models of $\O$). 
\begin{definition}
  \normalfont $Q_{g}$ is \emph{contained} in $Q'_{g}$ with respect to
  $\O$, denoted $Q_{g}\subseteq_{\O}Q'_{g}$, if for every model
  $\B\in\MOD(\O)$ we have that ${Q_{g}}^{\B}\subseteq{Q'_{g}}^{\B}$.
  $Q_{g}$ is \emph{proper contained} in $Q'_{g}$ with respect to $\O$,
  denoted $Q_{g}\subset_{\O}Q'_{g}$, if $Q_{g}\subseteq_{\O}Q'_{g}$
  and for at least one model $\B\in\MOD(\O)$ we have that
  ${Q_{g}}^{\B}\subset{Q'_{g}}^{\B}$.
\end{definition}
\begin{definition}
  \normalfont $Q_{g}$ is an \emph{optimal complete s-to-t rewriting}
  of $Q_{s}$ with respect to $\I$, if $Q_{g}$ is a complete s-to-t
  rewriting of $Q_{s}$ with respect to $\I$, and there exists no query
  $Q'_{g}$ such that $Q'_{g}$ is a complete s-to-t rewriting of
  $Q_{s}$ with respect to $\I$ and $Q'_{g}\subset_{\O}Q_{g}$.
\end{definition}
We are ready to present an algorithm for computing an optimal complete
s-to-t rewriting of a query over the source schema.
\begin{algorithm}
  \caption{FindOptimalComplete($\I$,
    $Q_{s}$)}\label{alg:FindUniqueOptimalComplete}
  \begin{algorithmic}[1]
    \Require{OBDM specification $\I=\tup{\O,\S,\M}$, CQ $Q_{s}$ over
      $\S$.}  

    \Ensure{query $Q_g$ over $\O$.}

    \BlankLine 

    \State Compute $\D_{Q_{s}}$ from $Q_{s}$ (i.e., the instance,
    possibly with incomplete information, associated to the query
    $Q_{s}$), and denote it with $\D$.

    \State Chase $\D$ w.r.t. $\M$ to produce an instance $J'$.  

    \State Chase $J'$ w.r.t. $\neg Q_{\sat}^{1\neq}$; if the chase
    fails, then \emph{stop} and return the query $\{\tuplehead(Q_{s})\
    |\ \bot(\tuplehead(Q_{s}))\}$; otherwise, let $J$ be the instance
    produced, and let $\psi$ be the set of equality applied to the
    variables in $\D$ by the chase.

    \State Evaluate $Q_{\sat}^{0\neq}$ over $J$; if the answer is
    $\{()\}$ (i.e., $J\models Q_{\sat}^{0\neq}$), then \emph{stop} and
    return the query $\{\tuplehead(Q_{s})\ |\
    \bot(\tuplehead(Q_{s}))\}$.  

    \State If $J=\emptyset$ (i.e., no atoms in the instance $J$), then
    \emph{stop} and return the query $\{\tuplehead(Q_{s})\ |\
    \top(\tuplehead(Q_{s}))\}$; otherwise, let $Q_{J}$ be the boolean
    conjunctive query associated to the instance $J$.  

    \State Let $\overline{w}$ be the tuple composed by all terms
    in $\psi(\tuplehead(Q_{s}))$ not appearing in $J$. If such tuple
    is not empty, then return $\{\psi(\tuplehead(Q_{s}))\ |\
    Q_{J}(\psi(\tuplehead(Q_{s}))\wedge\top(\overline{w})\}$;
    otherwise, return $\{\psi(\tuplehead(Q_{s}))\ |\
    Q_{J}(\psi(\tuplehead(Q_{s})))\}$.
	\end{algorithmic}
\end{algorithm}
For the termination and the complexity of this algorithm hold the same considerations done for the termination and the complexity of the \textsf{CheckComplete} algorithm. In particular, \textsf{FindOptimalComplete}($\I$,$Q_{s}$) terminates, and it runs in \myi $\PTIME$ in the size of $Q_{s}$; \myii $\PTIME$ in the size of $\O$; \myiii $\EXPTIME$ in the size of $\M$. Whereas, the correctness is established by the next theorem.
\begin{theorem} \textsf{FindOptimalComplete}($\I$, $Q_{s}$) returns an optimal complete s-to-t rewriting of $Q_s$ w.r.t. $\I$.
\end{theorem}
\begin{proof-sketch}
When the algorithm returns the query $\{\tuplehead(Q_{s})\ |\ \bot(\tuplehead(Q_{s}))\}$, it is easy to see that, regardless of which is the query $Q_{g}$, if we run the algorithm \textsf{CheckComplete}($\I$,$Q_{s}$,$Q_{g}$) it returns \emph{true} (also in this case, either the chase will fail, or the ABox $\A_{\D,\Sigma}$ produced will satisfy $Q_{\sat}^{0\neq}$), and hence, by Theorem \ref{thm:CheckComplete}, $Q_{g}$ is a complete s-to-t rewriting of $Q_{s}$ w.r.t. $\I$. It follows that, also $\{\tuplehead(Q_{s})\ |\ \bot(\tuplehead(Q_{s}))\}$ is a complete s-to-t rewriting, and, by definition of such query, it is an optimal complete s-to-t rewriting of $Q_{s}$ w.r.t. $\I$.

When the algorithm returns the query $Q_{g}=\{\psi(\tuplehead(Q_{s}))\ |\ Q_{J}(\psi(\tuplehead(Q_{s}))\wedge\top(\overline{w})\}$ (or $\{\tuplehead(Q_{s})\ |\ \top(\tuplehead(Q_{s}))\}$, in the case $J=\emptyset$), if we run the algorithm \textsf{CheckComplete}($\I$,$Q_{s}$,$Q_{g}$), it computes the ABox $\A_{\D,\Sigma}$, where $\tuplehead(Q_{s})\in\cert(Q_{g},\tup{\O,\A_{\D,\Sigma}})$ holds because $Q_{g}$ corresponds exactly to $\A_{\D,\Sigma}$ (before to be freezed) extended with $\top(\overline{w})$ for all terms $\overline{w}$ in $\psi(\tuplehead(Q_{s}))$ not appearing in $\A_{\D,\Sigma}$. It follows that, also in this case, \textsf{CheckComplete}($\I$,$Q_{s}$,$Q_{g}$) returns \emph{true}, implying, from Theorem \ref{thm:CheckComplete}, that $Q_{g}$ is a complete s-to-t rewriting of $Q_{s}$ w.r.t. $\I$. 

We now prove that the query $Q_{g}=\{\psi(\tuplehead(Q_{s}))\ |\ Q_{J}(\psi(\tuplehead(Q_{s}))\wedge\top(\overline{w})\}$ (or $\{\tuplehead(Q_{s})\ |\ \top(\tuplehead(Q_{s}))\}$, in the case $J=\emptyset$) is also an optimal complete s-to-t rewriting of $Q_{s}$ w.r.t. $\I$. In particular, suppose that there exist a query $Q'_{g}$ such that $Q'_{g}\subset_{\O}Q_{g}$, i.e., $Q'_{g}\subseteq_{\O}Q_{g}$, and there is a model $\B\in\sem^{\C}(\I)$ and a tuple $t$ such that $t\in Q_{g}^{\B}$ and $t\not\in {Q'_{g}}^{\B}$. The fact that $t\in Q_{g}^{\B}$ implies the existence of a valuation $v$ to all the variables in $Q_{g}$ that makes $Q_{g}$ true in $\B$. Note that, we can extend the valuation $v$ by assigning a new fresh constant to every variable appearing in $\D$ and not appearing in $Q_{g}$. The valuation $v$ obtained is now a valuation for $\D$, and obviously $t\in Q_{s}^{v(\D)}$. Moreover, if we apply the same valuation $v$ to the instance $J$, it is easy to see that we obtain a ground instance $J'$ such that $(v(\D),J')\models\M$ (we recall that $Q_{g}$ is the CQ associated to the instance $J$). Obviously, $J'\subseteq\B$, and hence, $(v(\D),\B)\models\M$ holds because queries in the mapping $\M$ are \emph{monotone} queries. Moreover, we also have that $\B\in\sem^{v(\D)}(\I)$ (the fact that $\B\models\O$ holds from the initial supposition). Hence, for the source database $v(\D)$ there is a model $\B\in\sem^{v(\D)}(\I)$ and a tuple $t$ such that $t\in Q_{s}^{v(\D)}$ and $t\not\in {Q'_{g}}^{\B}$, implying that, $Q'_{g}$ is not a complete s-to-t rewriting of $Q_{s}$ w.r.t. $\I$. \qed
\end{proof-sketch}
It is easy to prove that the query returned by the algorithm is not only an optimal complete s-to-t rewriting of $Q_{s}$ w.r.t. $\I$, but it is also the unique (up to equivalence) optimal complete s-to-t rewriting of $Q_{s}$ w.r.t. $\I$. Furthermore, the above result implies that an optimal complete s-to-t rewriting of $Q_{s}$ w.r.t. $\I$ can always be expressed as a CQ.

\section{Conclusion}

We have introduced the notion of Ontology-based Open Data Publishing,
whose idea is to use an OBDM specification as a basis for carrying out
the task of publishing high-quality open data. 

In this paper, we have focused on the bottom-up approach to
ontology-based open data publishing, we have introduced the notion of
source-to-target rewriting, and we have developed algorithms for two
problems related to complete source-to-target rewritings, namely the
recognition and the finding problem. We plan to continue our work on
several directions. In particular, we plan to investigate the notion
of sound rewriting under different semantics. Also, we want to study
the top-down approach, especially with the goal of devising techniques
for deriving which intensional knowledge to associate to datasets in
order to document their content in a suitable way.

\newpage
\bibliographystyle{abbrv}
\bibliography{bibliography}

\end{document}